\documentclass[aps,groupedaddress,superscriptaddress,amssymb,amsmath,amsbsy,amsfonts,twocolumn,prl]{revtex4}

\usepackage{amssymb,amsmath,amsthm,color,graphicx,times,graphicx}
\usepackage{hyperref}
\usepackage{enumerate}
\usepackage{subfigure}
\usepackage{dsfont}
\usepackage{times}
\usepackage{txfonts}
\usepackage{bbold}

\theoremstyle{plain}

\newtheorem{theorem}{Theorem}
\newtheorem{lemma}[theorem]{Lemma}

\theoremstyle{definition}

\date{\today}

\begin{document}

\title{Non-Markovianity Hierarchy of Gaussian Processes and Quantum Amplification}

\author{Pietro Liuzzo-Scorpo}
\affiliation{Centre for the Mathematics and Theoretical Physics of Quantum Non-Equilibrium Systems, \\ $\mbox{School of Mathematical Sciences, The University of Nottingham,
University Park, Nottingham NG7 2RD, United Kingdom}$}

\author{Wojciech Roga}
\affiliation{Department of Physics, University of Strathclyde, John Anderson Building, 107
Rottenrow, Glasgow G4 0NG, United Kingdom}

\author{Leonardo A. M. Souza}
\affiliation{Universidade Federal de Vi\c{c}çosa -- Campus Florestal, LMG818 Km6, Minas Gerais, Florestal 35690-000, Brazil}
\affiliation{Centre for the Mathematics and Theoretical Physics of Quantum Non-Equilibrium Systems, \\ $\mbox{School of Mathematical Sciences, The University of Nottingham,
University Park, Nottingham NG7 2RD, United Kingdom}$}

\author{Nadja K. Bernardes}
\affiliation{Departamento de F\'isica, Universidade Federal de Minas Gerais, Belo Horizonte, Caixa Postal 702, 30161-970, Brazil}

\author{Gerardo Adesso}
\affiliation{Centre for the Mathematics and Theoretical Physics of Quantum Non-Equilibrium Systems, \\ $\mbox{School of Mathematical Sciences, The University of Nottingham,
University Park, Nottingham NG7 2RD, United Kingdom}$}

\begin{abstract}
We investigate dynamics of Gaussian states of continuous variable systems under Gaussianity preserving channels. We introduce a hierarchy of such evolutions encompassing Markovian, weakly and strongly non-Markovian processes, and provide simple criteria to distinguish between the classes, based on the degree of positivity of intermediate Gaussian maps. We present an intuitive classification of  all one-mode Gaussian channels according to their non-Markovianity degree, and show that weak non-Markovianity has an operational significance as it leads to temporary phase-insensitive amplification of Gaussian inputs beyond the fundamental quantum limit. Explicit examples and applications are discussed.
\end{abstract}


\maketitle

\paragraph*{Introduction.}


Non-Markovian evolutions of open quantum systems have been extensively studied in recent years \cite{rivas2014,BreuerRMP,Alonso2016}. During these evolutions, memory effects appear in many forms 
\cite{rivas2014,BreuerRMP,wolf2008,breuer2009,rivas2010,lu2010,Laine2010,lorenzo2013,bylicka2014,Dhar2015,Souza2015,Buscemi2016,torre2015,Chruscinski2014,Modi2016}.
These effects can lead to enhancements in quantum computation, e.g.~for error correction or decoherence suppression \cite{laine2014,He2011,Biercuk2009,myatt2000,frank2009,Tong2010,Lue2013,haseli2014}, in quantum cryptography \cite{Vasile2011}, limiting the information accessible to the eavesdropper, and possibly in the efficiency of certain processes at the intersection between quantum physics and biology \cite{thorwart2009,chin2013,huelga2013}. Experimental techniques are now mature to investigate open quantum systems beyond the Markovian regime \cite{Liu2011,Chiuri2012,Tang2012,Xu2013,Liu2013,Fanchini2014,Orieux2015,Jin2015,Eisert2015,Bernardes2015,Bernardes2016}.

A quantum process defined by a completely positive (CP) dynamical map $\Lambda_t$ is Markovian if it is CP-divisible, i.e.~such that  an intermediate map $\tilde{\Lambda}_{t+\tau,t}$, defined by $\Lambda_{t+\tau}=\tilde{\Lambda}_{t+\tau,t}\Lambda_t$,
is CP for all  $t,\tau>0$. A CP map can indeed be represented by an interaction of the evolving system with an uncorrelated environment \cite{Stinespring1955}: lack of correlations at each step denotes lack of memory, hence Markovianity. On the other hand, we recognize a non-Markovian process when its description cannot be found among CP-divisible maps. In this case, correlations between system and environment are essential at some stage.

The association of Markovian processes with CP-divisible maps results in important restrictions. For instance, entanglement,  mutual information, or quantum channel capacity, cannot increase if a CP map is applied locally to the subsystems. Similarly, measures of state distinguishability, like fidelity or trace distance, are contractive under CP maps. A violation of CP-divisibility is then witnessed by the temporary increase of these quantities \cite{breuer2009,rivas2010,lorenzo2013,bylicka2014,lu2010,Laine2010,Dhar2015,Souza2015,Buscemi2016}. Proper measures of non-Markovianity rely on direct examination of complete positivity of all intermediate maps \cite{rivas2010,torre2015,Chruscinski2014,Modi2016}. A unified picture of several quantifiers of non-Markovianity has been presented in \cite{Chruscinski2014}, where a hierarchy of non-Markovianity degrees was introduced, based on the smallest degree of positivity of intermediate maps.

Further important insight into non-Markovian processes is achieved considering evolutions of quantum states living in infinite-dimensional Hilbert spaces, such as states of light. These states are described by continuous variables (CV) related to quadratures (position and momentum operators) \cite{ournewreview}. In the CV formalism, complete positivity of maps amounts to fulfilment of the uncertainty principle for any legitimate input state. Clearly, monotonicity of distinguishability or entanglement
under CP-divisible maps extends to  CV systems as well.

In this Letter we consider evolutions of Gaussian quantum states governed by Gaussianity preserving processes (Gaussian maps). Taking inspiration from the line drawn in \cite{Chruscinski2014} for finite-dimensional processes, and building on the methods of \cite{torre2015} where an elegant characterization of (non-)\!\!~Markovianity was accomplished for Gaussian maps in terms of CP-divisibility, we identify a simple general hierarchy based on the divisibility degree of Gaussian maps with Gaussian inputs. This hierarchy, obtained by introducing an intermediate notion of Gaussian $k$-positivity and providing necessary and sufficient criteria for it, allows us to distinguish three classes of processes: Markovian, weakly and strongly non-Markovian.

We then classify all one-mode Gaussian channels according to their non-Markovianity degree, using an intuitive pictorial diagram divided in three regions, one per each class of the hierarchy. We relate the latter to possibilities and limitations for quantum amplification. In particular, weakly non-Markovian phase-insensitive channels allow, during some intermediate time, for amplification of Gaussian inputs with less added noise  than the fundamental quantum limit \cite{Haus1962,Caves1982,Clerk2010,Eleftheriadou2013}. This provides a fascinating operational interpretation for weak non-Markovianity, which had gone unnoticed before \cite{Chruscinski2014}. 

\paragraph*{Gaussian states and Gaussian maps.}
Given an $n$-mode CV system, Gaussian states $\rho$ are defined as those having a Gaussian characteristic function in phase space \cite{adesso2007,ournewreview,paris2005}. These states are fully characterized by the first and second statistical moments of their quadrature vector $\hat{O}=\{\hat{q}_1,\hat{p}_1,\dots,\hat{q}_n,\hat{p}_n\}$, where $\hat{q}_j$ and $\hat{p}_k$ are canonically conjugate coordinates satisfying  $[\hat{q}_j,\hat{p}_k]=\mbox{$\frac{i}{2}$}\delta_{jk}$ (in natural units, $\hslash=1$). The first moment vector $D=\langle\hat{O}\rangle$ is also called the displacement vector, while the second moments ${(\sigma_n)}_{jk}=\frac12\langle\{\hat O_j, \hat O_k\}_+\rangle$ form the covariance matrix $\sigma_n$. All physical states must satisfy the Robertson-Schr\"{o}dinger uncertainty relation, $\sigma_n\geq\mbox{$\frac{i}{2}$}\Omega_n$, where $\Omega_n = {{\ 0 \ \ 1} \choose {-1 \ 0}}^{\oplus n}$ is the $n$-mode symplectic matrix \cite{Simon94}.

Quantum channels that preserve Gaussianity of their inputs are known as Gaussian maps. A Gaussian map acting on $n$-mode Gaussian states is represented by a pair of $2n \times 2n$ matrices $(X,Y)$, with $Y$ symmetric, acting on the displacement vector $D$ and the covariance matrix $\sigma_n$ as follows \cite{nogo1,nogo2,nogo3,torre2015,Lindblad2002,heinosaari2010},
\begin{equation}
D \ \  \rightarrow \ \ D'=XD\,, \qquad  \sigma_n \ \  \rightarrow \ \ \sigma'_n=X\sigma_n X^T+Y\,. \label{xsigy}
\end{equation}
A Gaussian map described by the pair $(X,Y)$ is CP if and only if the following well-known inequality is fulfilled \cite{torre2015,Lindblad2002,heinosaari2010,DePalma2015}
\begin{equation}\label{eq:cp}
Y-\mbox{$\frac{i}{2}$}\Omega_n+\mbox{$\frac{i}{2}$}X\Omega_nX^T\geq0~.
\end{equation}
This inequality can be obtained from the Stinespring dilation theorem \cite{Stinespring1955}, i.e.~considering the Gaussian map as the result of a Gaussian unitary evolution acting on system and environment, initialized in an uncorrelated $(n+m)$-mode Gaussian state, followed by partial trace over the $m$ environment modes, $\sigma_n \rightarrow \sigma'_n = {\rm Tr}_E \left[S ( \sigma_n \oplus \sigma_m^E ) S^T\right]$; here one uses the fact that a Gaussian unitary is represented by a symplectic transformation $S \in \text{Sp}(2(n+m),\mathbb{R})$ (i.e., one that preserves the symplectic matrix $\Omega$) acting by congruence on covariance matrices~\cite{ournewreview}.

\paragraph*{$k$-positivity of Gaussian maps.}
We now  introduce a notion of $k$-positivity for Gaussian maps with Gaussian inputs, inspired by the hierarchy of $k$-positivity for finite-dimensional  channels arising from Choi's theorem \cite{Choi1975}. We define a Gaussian map acting on $n$-mode Gaussian inputs as $k$-positive ($k$P) if its extension on $k$ additional modes is positive, i.e.~if, for all $(n+k)$-mode Gaussian states described by covariance matrices $\sigma_{n+k}\geq \frac{i}{2}\Omega_{n+k}$, it holds
\begin{equation}\label{eq:kpos}
(X\oplus\mathds{1}_k)\sigma_{n+k}(X\oplus\mathds{1}_k)^T+Y\oplus\mathbb{0}_k\geq \mbox{$\frac{i}{2}$}\Omega_{n+k}~.
\end{equation}

Interestingly, we prove in the Appendix \cite{epaps} that a Gaussian map with Gaussian inputs is CP if and only if it is $k$P with any $k \geq 1$. Precisely, we establish the following result.
\begin{theorem}\label{thm1}
For any $n$, the CP condition (\ref{eq:cp}) is equivalent to the $k$P condition (\ref{eq:kpos}) with $k=1$.
\end{theorem}
This means that, in the Gaussian scenario (unlike the general finite-dimensional case \cite{notefoot}), one has a very simple hierarchy of $k$-positivity, consisting of only three classes: completely positive (CP, $k=1$), positive (P, $k=0$) and not positive (NP) Gaussian maps. 
We can derive a simple (and, to our knowledge, original) condition to distinguish between the latter two classes, in terms of the pair $(X,Y)$. Noting that for  (\ref{eq:kpos}) to hold it suffices to check its validity on pure Gaussian states, whose covariance matrix can always be written as $\sigma=\frac12 SS^T$ with $S$ a symplectic transformation, we find that a Gaussian map with Gaussian inputs is positive ($k=0$) if and only if
\begin{equation}\label{eq:p}
   \quad \mbox{$\frac{1}{2}$}XSS^TX^T+Y-\mbox{$\frac{i}{2}$}\Omega_n\geq0\,, \quad \forall~S\in \text{Sp}(2n,\mathbb{R})~.
\end{equation}
Conditions (\ref{eq:cp}) and (\ref{eq:p}) allow one to fully classify the positivity properties of any $(n \rightarrow n)$-mode Gaussian map described by the pair $(X,Y)$ acting on Gaussian inputs. The conditions can be easily generalized to $(n \rightarrow m)$-mode Gaussian maps.

\paragraph*{Hierarchy of Gaussian non-Markovianity.}
Gaussian processes which are continuous in time are represented by a pair of time-dependent matrices $(X_t,Y_t)$ acting as in (\ref{xsigy}). Since we are interested in the divisibility properties of these maps, we can follow the approach of \cite{torre2015} and study the positivity of
the intermediate map $(X_\tau(t),Y_\tau(t))$ acting on the evolving system between times $t$ and $t+\tau$ and affecting the covariance matrix as usual, $\sigma(t)\rightarrow\sigma(t+\tau)=X_\tau(t)\sigma(t) X_\tau^T(t)+Y_\tau(t)$, with \cite{torre2015}
\begin{equation}
  X_\tau(t)=X_{t+\tau}X^{-1}_t~,\label{eq:intermediateXY} \quad
  Y_\tau(t)=Y_{t+\tau}-X_\tau(t)Y_tX_\tau^T(t)~.
\end{equation}

Our aim is now to provide a complete (non-)\!\!~Markovianity hierarchy of  Gaussian maps.
We first recall that imposing complete positivity of the intermediate map for all $t,\tau>0$, one obtains the condition for a Markovian evolution as  in \cite{torre2015},
\begin{equation}
  \mbox{$\frac{i}{2}$}X_\tau(t)\Omega X_\tau^T(t)+Y_\tau(t)-\mbox{$\frac{i}{2}$}\Omega\geq  0~.\label{cpco}
\end{equation}
Any Gaussian map not complying with condition (\ref{cpco}) at some intermediate times is non-Markovian \cite{torre2015}.
We can now add an extra layer to such a dichotomic characterization. Namely, if not all intermediate maps are CP, i.e.~if for some times condition (\ref{cpco}) is violated, but the positivity condition
\begin{equation}
  \mbox{$\frac{1}{2}$}X_\tau(t)SS^TX_\tau^T(t)+Y_\tau(t)- \mbox{$\frac{i}{2}$}\Omega\geq  0\,, \quad \forall~S\in \text{Sp}(2n,\mathbb{R})~.\label{pco}
\end{equation}
holds for all the maps, then the evolution is said to be \emph{weakly} non-Markovian. Finally, if there is at least one intermediate map violating  (\ref{pco}), the process is \emph{strongly} non-Markovian.

It is worth noting that, to check  CP- (respectively~P-) divisibility of the map $(X_t,Y_t)$, it suffices to verify that inequality (\ref{cpco}) [resp.~(\ref{pco})] holds in the limit of small $\tau$, since the composition of an arbitrary number  of intermediate CP (P) maps is CP (P).


\paragraph*{Complete classification of one-mode Gaussian maps.}
In what follows, we focus on one-mode quantum Gaussian processes. From the global map $(X_t,Y_t)$, we can construct, thanks to (\ref{eq:intermediateXY}), the intermediate maps given by the pairs of $2\times2$ matrices $(X_\tau(t),Y_\tau(t))$. In the limit of small $\tau$, $X_\tau(t)$ and $Y_\tau(t)$ are close to the identity and to the null matrix respectively. Expanding these matrices up to first order in $\tau$ we get
\begin{equation}
  X_\tau(t)=\left(1+\epsilon_t\tau\right)~\mathds{1}+\tau~{\cal X}(t)+o(\tau^2), \quad   Y_\tau(t) = \tau~{\cal Y}(t)+o(\tau^2),\label{Xseries}
\end{equation}
where ${\cal X}(t)$ and ${\cal Y}(t)$ are  arbitrary real matrices, with  ${\cal Y}(t)$ being symmetric.
The following two Theorems (proofs in the Appendix \cite{epaps}) then completely characterize the degree of Gaussian (non-)\!\!~Markovianity of any one-mode Gaussian map given by $(X_t,Y_t)$, in terms of the three real parameters
    \begin{eqnarray}
    \epsilon_t&\equiv& \mbox{$\frac{{d}~}{{d}t}$}\ln\left(\sqrt{|\det X_t|}\right)~,\label{eq:eps}\\
    \delta_t&\equiv&\left(\det X_t\right)^2\det \left( \mbox{$\frac{{d}~}{{d}t}$}\left(X^{-1}_tY_tX^{-T}_t\right)\right)~,\label{eq:del}\\
    \kappa_t&\equiv&\mbox{$\frac{{d}~}{{d}t}$}\mbox{tr}~Y_t- 2~\mbox{tr}\left(Y_t\mbox{$\frac{{d}~}{{d}t}$}\ln|X_t|\right)~.\label{eq:trace}
  \end{eqnarray}

\begin{theorem}\label{propom}
  A one-mode Gaussian process given by $(X_t,Y_t)$ is CP-divisible if, for all $t>0$,  it holds:
 $\delta_t \geq \epsilon_t^2$ and  $\kappa_t\geq 0$.
\end{theorem}


\begin{theorem}\label{propok}
  A one-mode Gaussian process given by $(X_t,Y_t)$ is divisible into positive intermediate maps (P-divisible) if, for all $t>0$, it holds:
   $ \delta_t \geq \mbox{$\frac14$} {(|\epsilon_t|-\epsilon_t)^2}$ and $\kappa_t \geq 0$.
\end{theorem}

The Gaussian processes for which Theorem~\ref{propom} is satisfied are Markovian. Those for which Theorem~\ref{propok} is satisfied while Theorem~\ref{propom} is not are weakly non-Markovian. Those for which Theorem~\ref{propok} is not satisfied are strongly non-Markovian. 



Let us now define
\begin{equation}\label{eq:mudef}
 \mu_t \equiv \left\{  \begin{array}{r@{\quad}cr}
\mathrm{sgn}(\kappa_t)\sqrt{\delta_t}~,\quad\mbox{for}~~\delta_t\geq0\\
-\sqrt{|\delta_t|}~,\quad\mbox{for}~~\delta_t<0
\end{array}\right.~.
\end{equation}
Due to Theorems \ref{propom} and \ref{propok}, for a one-mode Gaussian process we can then distinguish  three regions in the space of parameters  $\epsilon$ and $\mu$ as shown  in Fig.~\ref{fi2}, which correspond to the intermediate map being respectively CP, P, and NP:
\begin{equation}\label{eq:regions}\left.
  \begin{array}{llll}
    \Upsilon_{\mbox{\scriptsize CP}}&\equiv&\{(\epsilon,\mu)~|~\mu\geq|\epsilon|\}, & \hbox{} \\
    \Upsilon_{\mbox{\scriptsize P}}&\equiv&\{(\epsilon,\mu)~|~2\mu\geq|\epsilon|-\epsilon\}, & \hbox{} \\
    \Upsilon_{\mbox{\scriptsize NP}}&\equiv&\mathbb{R}^2\backslash\Upsilon_P. & \hbox{}
  \end{array}
\right.
\end{equation}

A similar diagram can be found e.g.~in \cite{Schaefer2013,Giovannetti2014,DePalma2015}. However, the parameters there characterize global quantum channels, so regions analogous to $\Upsilon_{\mbox{\scriptsize P}\backslash \mbox{\scriptsize CP}}\equiv\Upsilon_{\mbox{\scriptsize P}}\backslash\Upsilon_{\mbox{\scriptsize CP}}$ and $\Upsilon_{\mbox{\scriptsize NP}}$ are denoted as non-physical. Here, since the diagram is built for intermediate maps of a globally CP process,  which by themselves do not need to be CP, these regions are permitted. 

We can  in fact fully classify the Gaussian (non-)\!\!~Markovianity degree of any one-mode Gaussian process by studying the paths $\Gamma_t\equiv\{(\epsilon_s,\mu_s)\}_{s=0}^t$ defined by its intermediate maps on the $(\epsilon,\mu)$ diagram of Fig.~\ref{fi2}. If an evolution is Markovian then the trajectory will be confined at all times in the $\Upsilon_{\mbox{\scriptsize CP}}$ region,  $\Gamma_t~\in~\Upsilon_{\mbox{\scriptsize CP}}\quad\forall~t>0$.  If at some times the trajectory trespasses in the $\Upsilon_{\mbox{\scriptsize P}\backslash\mbox{\scriptsize CP}}$ region but never trespasses in the $\Upsilon_{\mbox{\scriptsize NP}}$ one, i.e.~if $\Gamma_t~\in~\Upsilon_{\mbox{\scriptsize P}}\quad\forall~t>0$ and $\exists~s : (\epsilon_s, \mu_s) \not\in \Upsilon_{\mbox{\scriptsize CP}}$, then the evolution is weakly non-Markovian. If at some times the trajectory crosses into the $\Upsilon_{\mbox{\scriptsize NP}}$ region, i.e.~$\exists~s : (\epsilon_s, \mu_s) \not\in \Upsilon_{\mbox{\scriptsize P}}$, then the evolution is strongly non-Markovian.

\paragraph*{Phase-insensitive maps: Allowed trajectories and examples.}
We now analyze in more detail the physical constraints imposed on processes described by  CP maps $(X_t,Y_t)$, transforming Gaussian states from an initial $t=0$ to a later time $t$.  
For ease of illustration, we will focus on the special case of phase-insensitive maps, which encompass the most physically relevant bosonic processes, such as quantum Brownian motion and amplitude damping \cite{Petruccione,schlosshauer,Vasile2011,Eisert2015,guarni2016,Souza2015,torre2015}. These have intermediate maps of the form $X_\tau(t) = \left(1+\epsilon_t\tau\right)~\mathds{1}$,   $Y_\tau(t)=\mu_t~\tau~\mathds{1}$, with $\mu_t\equiv\mbox{sgn}(\kappa_t)\sqrt{\delta_t}$, obtained by setting  ${\cal X}(t)=\mathbb{0}$ and ${\cal Y}(t) = \mu_t \mathds{1}$ in  (\ref{Xseries}).
Applying the composition law for Gaussian maps, 
it is easy to show that the global map from  $t=0$ to  $t=N\tau$, such that $\lim_{N\to\infty}\lim_{\tau\to0}N\tau=t>0$, is:
\begin{equation}
  X_t = e^{\int_0^t\epsilon_s ds}~\mathds{1}, \quad  Y_t=\mbox{$\Big(e^{2\int_0^t\epsilon_s ds}\int_0^t\mu_r e^{-2\int_0^r\epsilon_s ds}dr\Big) \mathds{1}$}\,.\label{XYt}
\end{equation}

\begin{figure}
\center
\vspace*{-.3cm}
\includegraphics[width=7.5cm]{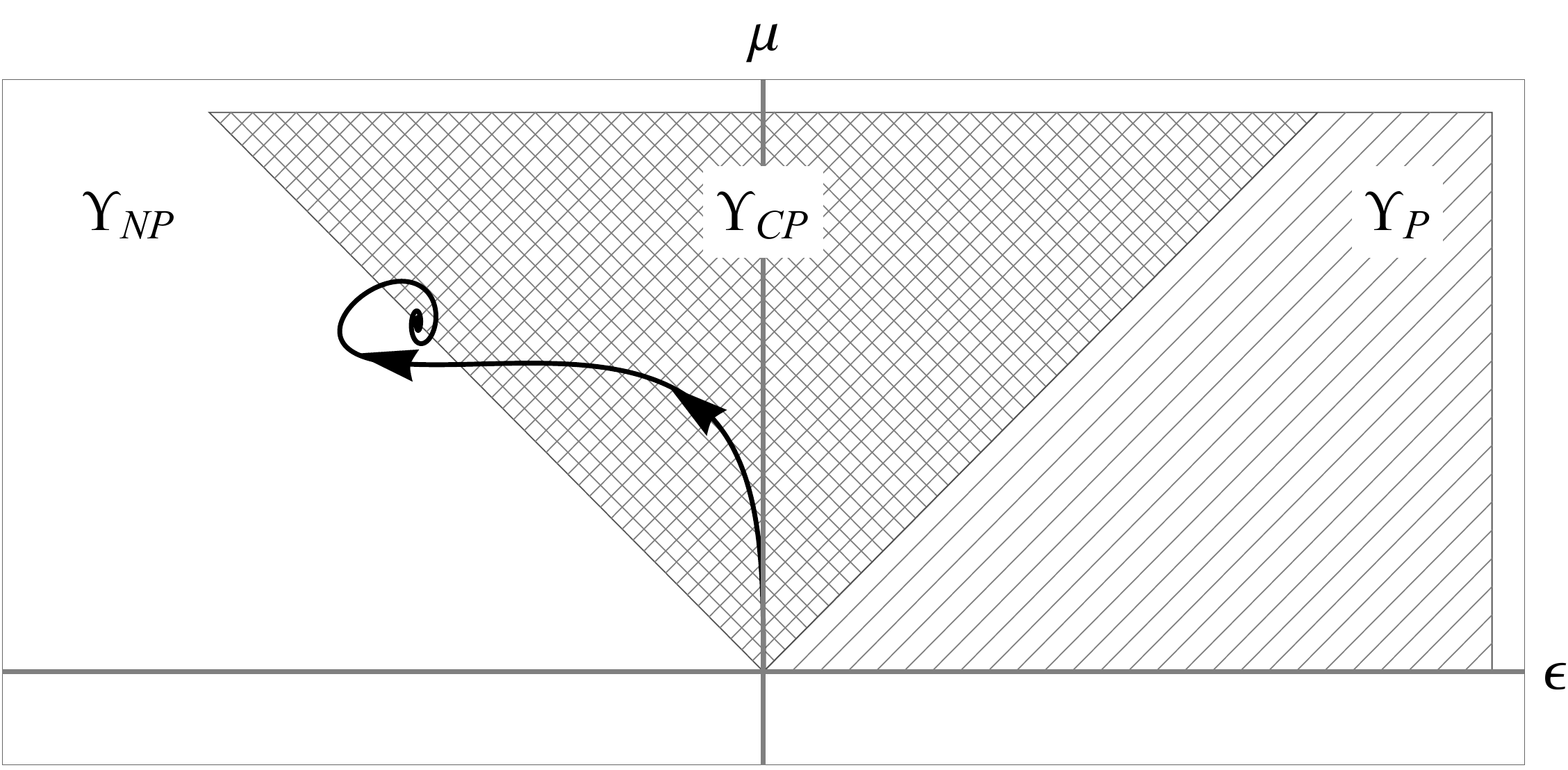}
\caption{Pictorial diagram of parameters $(\epsilon,\mu)$ characterizing one-mode Gaussian intermediate maps. The diagonal striped pattern corresponds to the P- but not CP-divisible region $\Upsilon_{\mbox{\scriptsize P}\backslash\mbox{\scriptsize CP}}$. The crosshatch pattern identifies the CP-divisible region $\Upsilon_{\mbox{\scriptsize CP}}$. The white region corresponds to $\Upsilon_{\mbox{\scriptsize NP}}$. A path on the diagram denotes a process with parameters changed continuously in time. The solid black path represents a quantum Brownian motion process as described in the text.
}
\label{fi2}
\end{figure}

A paradigmatic and widely studied example (see e.g.~\cite{Petruccione,schlosshauer} and references therein) is  quantum Brownian motion. With a secular and weak-coupling approximation, the master equation is given by: $
\dot{\rho}_t=\frac{\Delta_t+\gamma_t}{2}[2\hat a\rho_t \hat a^\dagger-\{\hat a^\dagger \hat a,\rho_t\}_+]+\frac{\Delta_t-\gamma_t}{2}[2\hat a^\dagger\rho_t \hat a-\{\hat a \hat a^\dagger,\rho_t\}_+]$, where $\hat{a}, \hat{a}^\dagger$ are the ladder operators satisfying $[\hat{a},\hat{a}^\dagger]=1$, while
$\Delta_t$ and $\gamma_t$ are respectively the diffusion and damping coefficients, which depend on the spectral density of the bath.
The evolved covariance matrix of a one-mode Gaussian state undergoing this dynamics is:
 $ \sigma(t) = \mbox{$\big(e^{-\int_0^t\gamma_s ds}\mathds{1}\big)\,\sigma(0)\,\big(e^{-\int_0^t\gamma_s ds}\mathds{1}\big)$}
  + \mbox{$e^{-2\int_0^t\gamma_s ds}\int_0^te^{2\int_0^s\gamma_r dr}\Delta_s ds~\mathds{1}$}$,
which corresponds to the map $(X_t,Y_t)$ given by (\ref{XYt}) with the substitutions $\epsilon_t\rightarrow-\gamma_t$,  $\mu_t\rightarrow\Delta_t$. A trajectory on the $(\epsilon,\mu)$ plane, for a system with characteristic frequency $\omega_0$ and a zero-temperature bath with Ohmic spectral density $J(\omega)=\omega e^{-\omega/\omega_c}$ and cut-off frequency $\omega_c=\omega_0/2$, is depicted in Fig.~\ref{fi2}.

More generally, to have a physical evolution from a composition of infinitesimal phase-insensitive maps, we must impose the CP condition (\ref{eq:cp}) on the global map $(X_t,Y_t)$ given by (\ref{XYt}). The eigenvalues of the lhs of (\ref{eq:cp})  are in this case:
  $\Lambda_\pm=\mbox{$\pm\mbox{$\frac{1}{2}$}+e^{2\int_0^t\epsilon_s ds}\left(\mp\mbox{$\frac{1}{2}$}+\int_0^te^{-2\int_0^r\epsilon_s ds}\mu_r dr\right)$}$.
The conditions $\Lambda_\pm\geq0$ can be rewritten as (see Appendix \cite{epaps})
\begin{equation}
\mbox{$  \int_0^te^{-2\int_0^s\epsilon_s ds}\left(\mu_r\pm\epsilon_r\right)dr\geq0,~\quad\forall~t>0$}~.\label{Lambda222}
\end{equation}
As expected, these conditions are weaker than the condition for CP-divisibility, allowing the trajectories in the diagram of Fig.~\ref{fi2} to go beyond the region $\Upsilon_{\mbox{\scriptsize CP}}$. However, the following constraint on the physical paths can be derived.
By expanding the lhs of  inequalities (\ref{Lambda222}) at first order in $t$ we get $ \mu_0\geq|\epsilon_0|$,
that is, the trajectory must begin in the CP region $\Upsilon_{\mbox{\scriptsize CP}}$. Moreover, if it starts on the boundaries of  $\Upsilon_{CP}$, i.e.,~$\mu_0=|\epsilon_0|$, then  $\dot{\mu}_0\geq|\dot{\epsilon}_0|$. This tells us that not only the trajectory must start in the CP-divisibility region, but it has to have an initial ``speed'' such that it will remain in there for the immediate subsequent time. A path which starts in the origin, then moves along the boundary of the crosshatched region up to a time $t_i$ and then trespasses in either region $\Upsilon_{\mbox{\scriptsize NP}}$ or $\Upsilon_{\mbox{\scriptsize P}\backslash\mbox{\scriptsize CP}}$, is not allowed.

\paragraph*{Operational significance of Gaussian non-Markovianity degrees.}
The  significance of the last no-go rule is related to fundamental physical properties.
Suppose indeed that at time $t=0$ an initial state is described by a thermal covariance matrix $\sigma=\mbox{diag}\{\nu,\nu\}$, with $\nu \geq \frac12$. Under the action of the map (\ref{XYt})  
at time $\tilde{t}>0$, the product of the canonical variances is
 $ \mbox{$\langle\hat{q}^2\rangle\langle\hat{p}^2\rangle=e^{4 \int_0^{\tilde{t}} \epsilon_s \, ds} \Big(\nu+  \int_0^{\tilde{t}} \mu_r e^{-2 \int_0^r \epsilon_s \,
   ds} \, dr\Big)^2$}$.
If $\Lambda_-< 0$, we obtain $\langle\hat{q}^2\rangle\langle\hat{p}^2\rangle <
e^{2 \int_0^{\tilde{t}} \epsilon_s \, ds} \Big(e^{ \int_0^{\tilde{t}} \epsilon_s \, ds}\nu-\sinh\big(\int_0^{\tilde{t}} \epsilon_s \,
   ds\big) \Big)^2$, which for a pure initial state (i.e.~the vacuum or a Glauber coherent state, with $\nu=\frac12)$ reduces to
 $ \langle\hat{q}^2\rangle\langle\hat{p}^2\rangle<\mbox{$\frac{1}{4}$}$,
i.e.~to a violation of the uncertainty principle, which is not physically admitted. Indeed, a trajectory lying along the border between $\Upsilon_{CP}$ and $\Upsilon_{NP}$ (representing, e.g., a damping master equation with generally time-dependent damping constant) preserves the purity of such a state. To better understand this, let us consider the limiting case of having a map such that $-\epsilon_t=\mu_t>0$ for $0<t<t_i$ and  $-\epsilon_t>\mu_t$ for $t_i<t<\tilde{t}$. Up to $t_i$, the $X_t$ part of the map decreases both variances of the pure input state, while the noise added by the $Y_t$ part compensates the loss and the state remains pure. Then, for $t>t_i$, the noise introduced by $Y_t$ is not enough and the uncertainty relation is violated.

However, crossing the border during the evolution would be possible if the preceding dynamics shrank the state domain of the intermediate map such that its subsequent action, corresponding to a temporary dilation of this domain, would not violate the uncertainty relation. The non-Markovian effect, manifested in the dilation of the volume of the physical states accessible during the dynamics, can then be seen as a backflow of information from the environment into the system \cite{lorenzo2013}.

Let us now comment on the other border of the CP region, between $\Upsilon_{\mbox{\scriptsize CP}}$ and $\Upsilon_{\mbox{\scriptsize P}\backslash\mbox{\scriptsize CP}}$. For any dynamics with added noise ($\mu_t  > 0$), a trajectory along this border is such that $\epsilon_t=\mu_t>0$, which is responsible for an {\em amplification}, that is, the multiplication of the displacement vector $D$ by a factor greater than $1$ and a corresponding increase of the variances. Along such a path, the noise added is the minimum allowed for quantum linear amplifiers \cite{Clerk2010}. Crossing this border into the $\Upsilon_{\mbox{\scriptsize P}\backslash\mbox{\scriptsize CP}}$ region at a time $t_i>0$ is allowed only if the noise added up to that time is sufficient to permit a subsequent amplification beyond the quantum limit. This is possible thanks to correlations established between system and environment during the preceding evolution. We can conclude therefore that a Gaussian phase-insensitive process (with added noise) is weakly non-Markovian if at any moment in time one observes that, although the covariances increase, a Gaussian state evolving under such a process is amplified beating the quantum limit. This provides an {\it operational interpretation} for the elusive phenomenon of weak non-Markovianity in the context of quantum amplification. 

\paragraph*{Conclusions.} This Letter introduced a meaningful hierarchy of non-Markovianity for CV Gaussian processes and established its physical significance. We provided a necessary and sufficient condition for positivity of a Gaussian map acting on Gaussian inputs. Applying this to intermediate maps, we then distinguished three main types of Gaussian processes: Markovian, weakly and strongly non-Markovian ones. 

In the one-mode case, we gave a simple prescription to identify to which class a Gaussian map belongs, based on its representation as a path $(\epsilon_t,\mu_t)$ in a two-dimensional diagram, where $\epsilon_t$ and $\mu_t$ can be computed explicitly from the pair of matrices $(X_t,Y_t)$ describing the action of the map. We also studied, in the physically relevant case of phase-insensitive channels, the constraints on these paths due to the requirement of having a global CP map. This allowed us to give a physical interpretation to weakly and strongly non-Markovian processes in terms of amplification beyond the quantum limit and of information backflow from the environment, respectively.

These findings can be of importance for quantum cryptography \cite{Vasile2011}. An eavesdropper with access to knowledge whether a given communication channel is weakly or strongly non-Markovian can amplify a state in such a way that the legitimate parties may find it too noisy to be useful, discarding it. Moreover, if the legitimate parties do not fully control the way the shared state is prepared, unexpected behaviour can be observed if possible non-Markovian effects are ignored.


We finally note that in all the Gaussian processes we considered explicitly (e.g.~quantum Brownian model and damping model), we found either instances of Markovian or strongly non-Markovian evolutions, but no weakly non-Markovian ones. This may be due to the fact that all these processes admit a final state at thermal equilibrium with the environment. Some purely weak non-Markovian processes might be retrieved in case an evolution in an active environment that pumps energy into the system is analyzed. Investigating memory effects in such processes deserves further investigation.


\begin{acknowledgments}
\paragraph*{Acknowledgments.} We thank Sabrina Maniscalco,  Dariusz Chru\'{s}ci\'{n}ski, John Jeffers, Marco Piani, Gianpaolo Torre, and Fabrizio Illuminati for discussions. This work was supported by the UK EPSRC Quantum Imaging Hub (Grant No.~EP/M01326X/1), the European Research Council (ERC) Starting Grant GQCOP (Grant No.~637352), the Foundational Questions Institute (fqxi.org)  Physics of the Observer Programme (Grant No.~FQXi-RFP-1601),  the Brazilian Agencies CAPES (Grant No.~6842/2014-03) and CNPq (Grant No.~470131/2013-6), and the University of Nottingham (Graduate School Travel Prize 2015).
\end{acknowledgments}


\begin{thebibliography}{63}
\bibitem{notefoot}
{A finite-dimensional channel acting on a $d$-dimensional system is $d'$-positive if its extension to the system plus a $d'$-dimensional ancilla is a positive map; then, said map is CP (i.e, $\infty$-positive) if and only if it is $d$-positive. Here, in the Gaussian scenario, we define instead $k$-positivity in terms of the extension to a number $k$ of ancillary CV modes. Since a single mode is already infinite-dimensional, it is perhaps  not surprising that complete positivity can be recovered already for $k=1$ in our framework.}
\end{thebibliography}

%

\clearpage
\appendix
\widetext
\begin{center}
\textbf{\large Supplemental Material}
\end{center}
\setcounter{page}{1}
\setcounter{equation}{0}
\setcounter{figure}{0}
\setcounter{table}{0}
\setcounter{section}{0}

\section{Proof of Theorem~\ref{thm1}}

In order to prove Theorem~\ref{thm1} it is useful to prove the following Lemma first.

\begin{lemma}\label{lem}
For any $2n\times 2n$ Hermitian matrix $H$ the smallest eigenvalue of $\frac{1}{2}~i\Omega_n+H$ is equal to the smallest eigenvalue of
\begin{equation}
\mbox{$\frac{1}{2}$}~i\Omega_{n-m}\oplus\mathds{1}_m +QHQ^T,
\label{withone}
\end{equation}
for some orthogonal symplectic matrix $Q$ and for any $m<n$.
\end{lemma}
\begin{proof}Let us denote $\lambda=\min\{{\rm eig}(H+\frac{1}{2}~i\Omega_n)\}$ which corresponds to an eigenvector
\begin{equation}
 v_{\lambda}=\begin{bmatrix}\alpha+i\beta\\  a_{n-1}+i b_{n-1}\\ a_n+i b_n\end{bmatrix},
\end{equation}
where $\alpha$ and $\beta$ are $2n-4$ dimensional real vectors and $a_{n-1},  a_n,  b_{n-1}$ and $ b_n$ are two dimensional real vectors. A transformation $Q\in \text{Sp}(2n,\mathds{R})\cap \text{SO}(2n)$ preserves the eigenvalues changing the corresponding eigenvector into $ v_\lambda'=Q v_\lambda$. In order to prove the Lemma we start showing that there exists $Q_1\in \text{Sp}(2n,\mathds{R})\cap \text{SO}(2n)$ such that $ v_\lambda^{(1)}=Q_1 v_\lambda$ is an eigenvector for both $\frac{1}{2}~i\Omega_n+Q_1HQ_1^T$ and $\frac{1}{2}~i\Omega_{n-1}\oplus\mathds{1}+Q_1HQ_1^T$ for the same eigenvalue $\lambda$. Denote
\begin{equation}
 v^{(1)}_{\lambda}=\begin{bmatrix}\alpha^{(1)}+i\beta^{(1)}\\  a^{(1)}_{n-1}+i b^{(1)}_{n-1}\\  a^{(1)}_n+i b^{(1)}_n\end{bmatrix}.
\end{equation}
Observe the action of $\frac{1}{2}~i\Omega_{n-1}\oplus\mathds{1}+Q_1HQ_1^T$ on $ v^{(1)}_\lambda$
\begin{eqnarray}
&&(Q_1HQ_1^T+\mbox{$\frac{1}{2}$}i\Omega_n+\mbox{$\frac{1}{2}$}\mathbb{0}_{n-1}\oplus(\mathds{1}-i\Omega))\begin{bmatrix}\alpha^{(1)}+i\beta^{(1)}\\ a^{(1)}_{n-1}+i b^{(1)}_{n-1}\\  a^{(1)}_n+i b^{(1)}_n\end{bmatrix}
=\lambda  v^{(1)}_\lambda+\mbox{$\frac{1}{2}$}\begin{bmatrix}\mathbb{0}\\ 0\\  a^{(1)}_n+\Omega b^{(1)}_n-i\Omega(a^{(1)}_n+\Omega b^{(1)}_n)\end{bmatrix}.
\label{abprim}
\end{eqnarray}
Consider the following symplectic orthogonal transformation
\begin{equation}
Q_1=\mathds{1}_{n-2}\oplus\begin{bmatrix}\cos\phi_1\mathds{1}&-\sin\phi_1  R_1\\
 \sin\phi_1  R_1& \cos\phi_1 \mathds{1}\end{bmatrix},
\label{esss}
\end{equation}
where
\begin{equation}
 R_1=\begin{bmatrix}\cos\theta_1 & -\sin\theta_1\\ \sin\theta_1 &\cos\theta_1\end{bmatrix}.
\end{equation}
Using this transformation we have
\begin{equation}
  v^{(1)}_\lambda=\begin{bmatrix}\alpha^{(1)}+i\beta^{(1)}\\  a^{(1)}_{n-1}+i b^{(1)}_{n-1}\\  a^{(1)}_n+i b^{(1)}_n\end{bmatrix}=Q_1\begin{bmatrix}\alpha \\  a_{n-1}\\  a_n\end{bmatrix}+iQ_1\begin{bmatrix}\beta \\ b_{n-1}\\  b_n\end{bmatrix}=\begin{bmatrix}\alpha\\ \cos\phi_1  a_{n-1}-\sin\phi_1 R_1 a_n\\ \cos\phi_1  a_n+\sin\phi_1  R_1 a_{n-1}\end{bmatrix}+i\begin{bmatrix}\beta\\ \cos\phi_1  b_{n-1}-\sin\phi_1 R_1 b_n\\ \cos\phi_1  b_n+\sin\phi_1  R_1 b_{n-1}\end{bmatrix}.
\end{equation}
The term $ a^{(1)}_n+\Omega b^{(1)}_n$ from the last term of (\ref{abprim}) can now be written as
\begin{equation}
 a^{(1)}_n+\Omega b^{(1)}_n=\cos\phi_1( a_n+\Omega b_n)+\sin\phi_1 R_1( a_{n-1}+\Omega b_{n-1}).
\label{eleven}
\end{equation}
Notice that for any two real two dimensional vectors $ r_1$ and $ r_2$ one can always find a rotation $ R_1$ and an angle $\phi_1$ such that $\cos\phi_1 r_1+\sin\phi_1  R_1 r_2=0$. Indeed, the rotation $ R_1$ directs the second vector to be parallel to the first and $\sin\phi_1$ and $\cos\phi_1$ adjust the lengths.
Therefore, we showed that it is possible to find a symplectic orthogonal transformation $Q_1$, i.e. $ R_1$ and $\phi_1$, such that the last term of (\ref{abprim}) vanish, hence that ${\bf v}^{(1)}_\lambda$ is an eigenvector of $\frac{1}{2}~i\Omega_{n-1}\oplus\mathds{1}+Q_1HQ_1^T$ corresponding to the eigenvalue $\lambda$.

Using an analogous argument we show that $\lambda$ is also an eigenvalue of $\frac{1}{2}~i\Omega_{n-2}\oplus\mathds{1}_2+Q_2\tilde{H}Q_2^T$, where $\tilde{H}=Q_1HQ_1^T$, corresponding to the eigenvector $ v^{(2)}_\lambda=Q_2 v^{(1)}_\lambda$, with
\begin{equation}
  Q_2=\mathds{1}_{n-3}\oplus\begin{bmatrix}\cos\phi_2\mathds{1}&-\sin\phi_2  R_2\\
 \sin\phi_2  R_2& \cos\phi_2 \mathds{1}\end{bmatrix}\oplus\mathds{1}~,
\end{equation}
and $\phi_2$, $ R_2$ satisfying
\begin{equation}
\cos\phi_2( a^{(1)}_{n-1}+\Omega b^{(1)}_{n-1})+\sin\phi_2  R_2( a^{(1)}_{n-2}+\Omega b^{(1)}_{n-2})=0~.
\end{equation}
Iterating this procedure we find that $Q=Q_k\cdot Q_{k-1}\dots Q_2\cdot Q_1$. This completes the proof of Lemma~\ref{lem}.
\end{proof}

\begin{proof}[Proof (of Theorem \ref{thm1})]

We want now to deliver a condition on a map $(X,Y)$ acting on an $n$-mode quantum system guaranteeing that the inequality (\ref{eq:kpos}) is satisfied for every $\boldsymbol{\sigma}_{n+k}$ where $1\leq k$. We consider a generic bipartite $n+k$-modes covariance matrix
\begin{equation}
\sigma_{n+k}=\left(\begin{array}{cc}
A & C\\
C^T & B\end{array}\right)
\end{equation}
where $A$ is a $2n\times2n$ symmetric matrix, $B$ is a $2k\times2k$ symmetric matrix and $C$ is a $2n\times2k$ matrix. Inequality (\ref{eq:kpos}) reads
\begin{equation}\label{eq:posA}
\left(\begin{array}{cc}
XAX^T+Y-\mbox{$\frac{i}{2}$}\Omega_n & XC\\
C^TX^T & B-\mbox{$\frac{i}{2}$}\Omega_k\end{array}\right) \geq0 \Leftrightarrow
\left(\begin{array}{cc}
A+X^{-1}(Y-\mbox{$\frac{i}{2}$}\Omega_n)(X^T)^{-1} & C\\
C^T & B-\mbox{$\frac{i}{2}$}\Omega_k\end{array}\right) \geq0,
\end{equation}
where we assume that $X$ is invertible. As $\sigma_{n+k}\geq \frac{i}{2}\Omega_{n+k}$, we also have that $B-\frac{i}{2}\Omega_k\geq0$. Assuming invertibility of $ B-\frac{i}{2}\Omega_k$ the condition (\ref{eq:posA}) is equivalent to positivity of the Schur's complement of the block $B-\mbox{$\frac{i}{2}$}\Omega_k$, i.e.
\begin{equation}\label{eq:posB}
A+X^{-1}(Y-\mbox{$\frac{i}{2}$}\Omega_n)(X^T)^{-1}-C(B-\mbox{$\frac{i}{2}$}\Omega_k)^{-1}C^T\geq0.
\end{equation}
Moreover, applying the Schur's complement Lemma to $\sigma_{n+k}-\frac{i}{2}\Omega_{n+k}$ we get
\begin{equation}
A-\mbox{$\frac{i}{2}$}\Omega_n-C(B-\mbox{$\frac{i}{2}$}\Omega_k)^{-1}C^T\geq0.
\end{equation}
Hence, the lhs of (\ref{eq:posB}) can be decomposed in a positive state-dependent term and in a map-dependent one:
\begin{equation}\label{eq:posC}
\underbrace{A-\mbox{$\frac{i}{2}$}\Omega_n-C(B-\mbox{$\frac{i}{2}$}\Omega_k)^{-1}C^T}_{\text{state-dependent}}+\underbrace{X^{-1}(Y-\mbox{$\frac{i}{2}$}\Omega_n)(X^T)^{-1}+\mbox{$\frac{i}{2}$}\Omega_n}_{\text{map-dependent}}\geq0~.
\end{equation}
This condition has to be satisfied for any $n+k$-modes state. Due to the Williamson's theorem we can derive that for any mixed state $\sigma_{\mbox{mixed}}$ there exists a pure state $\sigma_{\mbox{pure}}^{S}$ such that
\begin{equation}
\sigma_{\mbox{mixed}}=S~\mbox{diag}\{\nu_1,\nu_1,\dots,\nu_{n+k},\nu_{n+k}\}~S^T\geq\mbox{$\frac{1}{2}$}S~\mathds{1}_{n+k}~S^T=\sigma_{\mbox{pure}}^{S}~. \end{equation}
It is then sufficient to check that inequality (\ref{eq:posC}) holds for pure states to guarantee that it is satisfied for all states. By using  again the Williamson's Theorem we find that local symplectic transformations $S_n$ and $S_k$ can bring the covariance matrix of any pure $n+k$-mode Gaussian state $(S_n\oplus S_k)\sigma^{S}_{\mbox{pure}}(S_n\oplus S_k)^T$ to the normal form i.e. the block form with non-zero entries only on the diagonal of each block. For $k\leq n$ the blocks are
\begin{align}
A&=\mbox{$\frac{1}{2}$}S_n\left(\bigoplus_{j=1}^k \cosh r_j \mathds{1}\oplus\mathds{1}_{n-k}\right)S_n^T~,\\
B&=\mbox{$\frac{1}{2}$}S_k\left(\bigoplus_{j=1}^k \cosh r_j~\mathds{1}\right)S_k^T,\\
C&=\mbox{$\frac{1}{2}$}S_n\left(\begin{array}{c}
\bigoplus_{j=1}^k \sinh r_j~\Lambda\\
\hline
\diamond\end{array}\right)S_k^T~,
\end{align}
where $\diamond$ is an $2(n-k)\times2n$ null matrix. We have then
\begin{equation}
C(B-\mbox{$\frac{i}{2}$}\Omega_k)^{-1}C^T=\mbox{$\frac{1}{2}$}S_n\left(\bigoplus_{j=1}^k\left(\cosh r_j\mathds{1}-i\Omega\right)\oplus\mathbb{0}_{n-k}\right)S_n^T~.
\end{equation}
In consequence, (\ref{eq:posC}) is satisfied for any state if
\begin{equation}\label{eq:prop}
\mbox{$\frac{1}{2}$}S_n\left(i\Omega_k\oplus\mathds{1}_{n-k}\right)S_n^T+X^{-1}(Y-\mbox{$\frac{i}{2}$}\Omega_n)(X^T)^{-1}\geq0
\end{equation}
holds for every $S_n\in \text{Sp}(2n,\mathds{R})$. Notice that for every $S_n$
\begin{equation}
\mbox{$\frac{1}{2}$}S_n(i\Omega_k\oplus\mathds{1}_{n-k})S_n^{T}+X^{-1}(Y-\mbox{$\frac{i}{2}$}\Omega_n)(X^T)^{-1}\geq \mbox{$\frac{i}{2}$}\Omega_n+X^{-1}(Y-\mbox{$\frac{i}{2}$}\Omega)(X^T)^{-1}.
\label{psmallercp}
\end{equation}
This inequality implies that the lhs cannot have an eigenvalue smaller than the smallest eigenvalue of the rhs. To complete the proof of Theorem \ref{thm1} it is sufficient to show that there exists $S_n$ such that the lhs and the rhs have the same the smallest eigenvalue for any $1\leq k$. For $1\leq k\leq n$, this is guaranteed by Lemma \ref{lem}. Indeed, this lemma shows that there exists a symplectic orthogonal transformation $Q$ such that
\begin{equation}
  \min\left\{\mbox{eig}\left(\mbox{$\frac{1}{2}$}Q^T(i\Omega_k\oplus\mathds{1}_{n-k})Q+X^{-1}(Y-\mbox{$\frac{i}{2}$}\Omega_n)(X^T)^{-1}\right)\right\}= \min\left\{\mbox{eig}\left(\mbox{$\frac{i}{2}$}\Omega_n+X^{-1}(Y-\mbox{$\frac{i}{2}$}\Omega)(X^T)^{-1}\right)\right\}~.
\end{equation}
If $k>n$, then (\ref{eq:posC}) becomes equal to the lhs of (\ref{psmallercp}). Summarizing, the positivity condition (\ref{eq:kpos}) for $k=1$ is equivalent to the positivity condition for any $k\geq 1$. This completes the proof of Theorem 1.
\end{proof}


Alternatively, the theorem can be justified by an extension of Choi's theorem on continuous variable systems, noting that a single mode is already an infinite dimensional system. A formal establishment of this argument can be easily derived.

\section{Proof of Theorem~\ref{propom}}

\begin{proof}
Let us start with simplifying the CP condition (\ref{cpco}): since $X_\tau(t)$ is a $2\times2$ matrix, we have that $X_\tau(t)\Omega X_\tau^T(t)=\Omega\det X_\tau(t)$. This allows us to reduce the CP-divisibility condition (\ref{cpco}) to the following form
\begin{equation}
Y_\tau(t)+\frac{i}{2}(\det X_\tau(t)-1)\Omega\geq 0~,
\label{ygt}
\end{equation}
Moreover, noticing that
\begin{equation}
  X_\tau(t)=\left(\begin{array}{cc}
  (1+\tau\epsilon)+\tau\delta & \tau\phi \\
  \tau\alpha & (1+\tau\epsilon)-\tau\delta\end{array}\right)+o(\tau^2)
\end{equation}
we have
\begin{eqnarray}
  \det X_\tau(t)&=&(1+\tau\epsilon)^2-\tau^2(\delta^2-\alpha\phi)+o(\tau^4)=\nonumber\\
  &=&1+2\tau\epsilon+o(\tau^2)~,
\end{eqnarray}
and making use of (\ref{eq:intermediateXY}), we find the expression for $\epsilon_t$ (\ref{eq:eps}):
\begin{eqnarray}
  \epsilon_t&=&\lim_{\tau\rightarrow0}\frac{\det X_\tau(t)-1}{2\tau}= \lim_{\tau\rightarrow0}\frac{\det X_{t+\tau}X_t^{-1}-\det X_tX_t^{-1}}{2\tau}=\nonumber\\
   &=&\lim_{\tau\rightarrow0}\left(\frac{\det X_{t+\tau}-\det X_t}{2\tau}\right)\det X_t^{-1}=\frac{1}{2}\frac{d\det X_t}{dt}\frac{1}{\det X_t} =\nonumber\\
  &=& \frac{d~}{dt}\ln\left(\sqrt{|\det X_t|}\right)\nonumber
\end{eqnarray}

Now, since $Y_{\tau}$ is a $2\times 2$ real symmetric matrix, it can be diagonalized by orthogonal transformations. Moreover, through a symplectic transformation of the form $Z_z\equiv\mbox{diag}\{z,1/z\}$ it can be brought to a diagonal form proportional to the identity or to the Pauli matrix $\sigma_z$. The proportionality factor is $\mu_t$ such that
\begin{equation}\label{eq:mu}
   \mu_t^2=\mbox{sgn}(\delta_t)\delta_t~,
 \end{equation}
where $\delta_t\equiv\det \mathcal{Y}(t)$. Making use of (\ref{eq:intermediateXY}) one finds the expression for $\delta_t$ (\ref{eq:del}):
\begin{eqnarray}
  \det \mathcal{Y}(t)&=&\lim_{\tau\rightarrow0}\det\left(\frac{Y_\tau(t)}{\tau}\right)= \lim_{\tau\rightarrow0}\det \left(\frac{Y_{t+\tau}-X_\tau(t)Y_tX_\tau^T(t)}{\tau}\right)=\nonumber\\
  &=&\lim_{\tau\rightarrow0}\det \left(\frac{Y_{t+\tau}-X_{t+\tau}X_t^{-1}Y_tX_t^{-T}X_{t+\tau}^T}{\tau}\right)=\nonumber\\
  &=&\lim_{\tau\rightarrow0}\det \left[\frac{X_{t+\tau}\left(X_{t+\tau}^{-1}Y_{t+\tau}X_{t+\tau}^{-T}-X_t^{-1}Y_tX_t^{-T}\right)X_{t+\tau}^T}{\tau}\right]\nonumber\\
  &=&\lim_{\tau\rightarrow0}\det \left(\frac{X_{t+\tau}^{-1}Y_{t+\tau}X_{t+\tau}^{-T}-X_t^{-1}Y_tX_t^{-T}}{\tau}\right) \left(\det X_{t+\tau}\right)^2= \nonumber\\
  &=&\left(\det X_t\right)^2\det\left(\frac{d~}{dt}\left(X_t^{-1}Y_tX_t^{-T}\right)\right)\nonumber
\end{eqnarray}

Let us consider the case $\delta_t\geq0$, i.e.~$\mathcal{Y}(t)$ is positive definite or negative definite, and let $\mu_t$ be the simplectic eigenvalue with the sign of $\mathcal{Y}(t)$, i.e.~$\mu_t=\mbox{sgn}(\kappa_t)\sqrt{\delta_t}$, where
\begin{eqnarray}
  \kappa_t&=&\mbox{tr}~\mathcal{Y}(t)=\lim_{\tau\rightarrow0}\frac{\mbox{tr}~ Y_\tau(t)}{\tau}= \lim_{\tau\rightarrow0}\frac{\mbox{tr} \left(Y_{t+\tau}-X_\tau(t)Y_tX_\tau^T(t)\right)}{\tau}=  \nonumber\\
  &=&\lim_{\tau\rightarrow0}\frac{\mbox{tr} \left[X_{t+\tau}\left(X_{t+\tau}^{-1}Y_{t+\tau}X_{t+\tau}^{-T}-X_t^{-1}Y_tX_t^{-T}\right)X_{t+\tau}^T\right]}{\tau}=\nonumber \\ &=&\mbox{tr}\left(\frac{d~}{dt}\left(X_t^{-1}Y_tX_t^{-T}\right)X_t^TX_t\right)= \nonumber\\
  &=&\mbox{tr}\left(\frac{d~}{dt}\left(X_t^{-1}Y_tX_t^{-T}X_t^{T}X_t\right)\right) -\mbox{tr}\left(X_t^{-1}Y_tX_t^{-T}\frac{d~}{dt}\left(X_t^TX_t\right)\right)=\nonumber\\
  &=&\frac{d~}{dt}\mbox{tr}~Y_t-\mbox{tr}\left(X_t^{-1}Y_tX_t^{-T}\frac{dX_t^T}{dt}X_t\right) -\mbox{tr}\left(X_t^{-1}Y_tX_t^{-T}X_t^T\frac{dX_t}{dt}\right)=\nonumber\\
  &=&\frac{d~}{dt}\mbox{tr}~Y_t-\mbox{tr}\left(Y_tX_t^{-T}\frac{dX_t^T}{dt}\right)^T -\mbox{tr}\left(X_t^{-1}Y_t\frac{dX_t}{dt}\right)=\nonumber\\
  &=&\frac{d~}{dt}\mbox{tr}~Y_t-2\mbox{tr}\left(\frac{dX_t}{dt}X_t^{-1}Y_t\right) =\frac{d~}{dt}\mbox{tr}~Y_t-2~\mbox{tr}\left(\frac{d\ln|X_t|}{dt}Y_t\right)~.\nonumber
\end{eqnarray}

Since $Y_\tau(t)$ can be brought into its diagonal form by a symplectic transformation, $\Omega$ is invariant under this transformation and it doesn't change the sign of the inequality, we can rewrite the CP condition as
\begin{equation}
\mu_t~\mathds{1}+i~\epsilon_t~\Omega+o(\tau)\geq 0.
\end{equation}
The CP (infinitesimal) divisibility condition can then be easily expressed in terms of $\mu_t$ and $\epsilon_t$:
\begin{equation}\label{eq:muepsCPcond}
\mu_t\geq |\epsilon_t|\qquad\forall~t\geq0~.
\end{equation}
It is obvious that if $\mu_t<0$, i.e.~if ${\bf Y}(t)$ is negative definite, the above condition is never satisfied. In the case $\delta_t<0$, with $\mu_t=\sqrt{|\delta_t|}$, through a symplectic transformation we can bring (\ref{ygt}) into the form
\begin{equation}
\pm\mu_t~\sigma_z+i~\epsilon_t~\Omega+o(\tau)\geq 0,
\end{equation}
which is never satisfied.
\end{proof}

\section{Proof of Theorem~\ref{propok}}

\begin{proof}
Exploiting again the property that any 2x2 matrix divided by the square root of its determinant is a symplectic matrix, the P-divisibility condition (\ref{eq:p}) can be rewritten as
\begin{equation}
\forall~~S\in \text{Sp}(2,\mathds{R})~,\quad \frac{\det X_\tau(t)}{2}SS^T+Y_\tau(t)-\frac{i}{2}\Omega\geq  0~.
\end{equation}
We first consider the case $\delta_t\geq0$, the above inequality can be recast as
\begin{equation}
\forall~~S\in \text{Sp}(2,\mathds{R})~,\quad \frac{\det X_\tau(t)}{2}SS^T+\mu_t\tau\mathds{1}-\frac{i}{2}\Omega\geq  0~.
\end{equation}

Using the Euler decomposition of symplectic transformations $S=O_1Z_zO_2$, where $Z_z\equiv\mbox{diag}\{z,1/z\}$ with $z\in(0,1]$ and $O_i$ is an orthogonal matrix we can further simplify the P condition as follows
\begin{equation}
\forall~z\in(0,1]~,\quad \frac{\det X_\tau(t)}{2}Z_z+\mu_t\tau\mathds{1}-\frac{i}{2}\Omega\geq  0~.
\label{forol}
\end{equation}
At first order in $\tau$, the eigenvalues of the lhs of (\ref{forol}) are
\begin{eqnarray}
  \lambda_1&=&\left(\frac{2 ~\epsilon_t~ z^2}{z^4+1}+\mu_t \right) \tau ~ ,\\
  \lambda_2&=&\frac{z^4+1}{2
   z^2}+\left(\frac{\left(z^8+1\right) \epsilon_t }{z^6+z^2}+\mu_t \right)~ \tau ~.
\end{eqnarray}
We notice that $\lambda_2$ is always positive for small $\tau$, hence the positivity of the intermediate map depends only on $\lambda_1$; in particular we have that the intermediate map is positive if
\begin{equation}
  \frac{2 \epsilon_t~ z^2}{z^4+1}+\mu_t\geq0\qquad\forall~~z\in(0,1]~,
\end{equation}
which is equivalent to
\begin{equation}\label{eq:muepsPcond}
    \mu_t\geq\frac{|\epsilon_t|-\epsilon_t}{2}~.
\end{equation}

In the case of $\delta_t<0$ the P condition becomes
\begin{equation}
\forall~z\in(0,1]~,\quad \frac{\det X_\tau(t)}{2}Z_z\pm\mu\tau\sigma_z-\frac{i}{2}\Omega\geq  0~,
\end{equation}
which is never satisfied.
\end{proof}

\section{Proof of Eq.~(\ref{Lambda222})}

The condition $\Lambda_+\geq0$ can be rewritten as
\begin{eqnarray}
  &1+e^{2\int_0^t\epsilon_sds}\left(-1+2\int_0^te^{-2\int_0^r\epsilon_sds}\mu_rdr\right)\geq0\nonumber\\
  &e^{-2\int_0^t\epsilon_sds}-1+2\int_0^te^{-2\int_0^r\epsilon_sds}\mu_rdr\geq0\nonumber\\
  &-e^{-\int_0^t\epsilon_sds}\sinh\left(\int_0^t\epsilon_sds\right)+\int_0^te^{-2\int_0^r\epsilon_sds}\mu_rdr\geq0\nonumber
\end{eqnarray}
Noticing that
\begin{eqnarray}
  e^{-\int_0^t\epsilon_sds}\sinh\left(\int_0^t\epsilon_sds\right)
  &=&\int_0^t\frac{d}{dt}\left[e^{-\int_0^t\epsilon_sds}\sinh\left(\int_0^t\epsilon_sds\right)\right]dt=\nonumber\\
  &=&\int_0^te^{-2\int_0^r\epsilon_rdr}\epsilon_rdr\nonumber
\end{eqnarray}
we finally get
\begin{equation}
  \int_0^te^{-2\int_0^r\epsilon_rdr}\left(\mu_r-\epsilon_r\right)dr\geq0~\quad\forall~t>0~.
\end{equation}

Analogously it can be shown that the condition $\Lambda_-\geq0$ is equivalent to
\begin{equation}
  \int_0^te^{-2\int_0^r\epsilon_rdr}\left(\mu_r+\epsilon_r\right)dr\geq0~\quad\forall~t>0~.
\end{equation}

\end{document}